\documentclass[%
 reprint,
 amsmath,amssymb,
pra,
]{revtex4-2}

\usepackage[dvipsnames]{xcolor}

\newtheorem{theorem}{Theorem}
\newtheorem{lemma}{Lemma}
\newtheorem{proof}{Proof}
\newtheorem{remark}{Remark}
\usepackage{braket}
\usepackage{graphicx} 
\usepackage{subcaption}
\usepackage{booktabs,array}
\usepackage{bbm}
\usepackage{calc} 
\usepackage{etoolbox}

\usepackage{tocloft}

\newcommand{\ii}{\mathrm{i}}
\newcommand{\e}{\mathrm{e}}

\usepackage{listings}
\usepackage[ruled]{algorithm2e}
\usepackage[colorlinks,linkcolor=blue,anchorcolor=blue,citecolor=blue]{hyperref}

\usepackage{graphicx}
\usepackage{dcolumn}
\usepackage{bm}

\begin{document}

\title{Quantum Walk Search on Complete Multipartite Graph with Multiple Marked Vertices}

\author{Ningxiang Chen}
\affiliation{State Key Lab of Processors, Institute of Computing Technology, Chinese Academy of Sciences, Beijing 100190, China}
\affiliation{School of Computer Science and Technology, University of Chinese Academy of Sciences, Beijing 100049, China}

\author{Meng Li}
\email[Corresponding author: ]{limeng2021@ict.ac.cn}
\affiliation{State Key Lab of Processors, Institute of Computing Technology, Chinese Academy of Sciences, Beijing 100190, China}
\affiliation{School of Computer Science and Technology, University of Chinese Academy of Sciences, Beijing 100049, China}

\author{Xiaoming Sun}
\affiliation{State Key Lab of Processors, Institute of Computing Technology, Chinese Academy of Sciences, Beijing 100190, China}
\affiliation{School of Computer Science and Technology, University of Chinese Academy of Sciences, Beijing 100049, China}

\date{\today}

\begin{abstract}
Quantum walk is a potent technique for building quantum algorithms. This paper examines the quantum walk search algorithm on complete multipartite graphs with multiple marked vertices, which has not been explored before. Two specific cases of complete multipartite graphs are probed in this paper, and in both cases, each set consists of an equal number of vertices. 
We employ the coined quantum walk model and achieve quadratic speedup with a constant probability of finding a marked vertex. Furthermore, we investigate the robust quantum walk of two cases and demonstrate that even with an unknown number of marked vertices, it is still possible to achieve a quadratic speedup compared to classical algorithms and the success probability oscillates within a small range close to 1.
This work addresses the overcooking problem in quantum walk search algorithms on some complete multipartite graphs.
We also provide the numerical simulation and circuit implementation of our quantum algorithm.
\end{abstract}

\maketitle

\section{Introduction}

Quantum computing is an approach that utilizes quantum mechanics to process information and perform calculations. In recent years, it has been extensively researched for its potential in solving various problems \cite{li2022quantum, zheng2022quantum, he2023exact, gao2023quantum}.
Quantum walk is a quantum version of the classical random walk, which is a fundamental concept in quantum computation and quantum information and was first introduced by Aharonov, Davidovich, and Zagury in 1993 \cite{aharonov_quantum_1993}. It has been widely studied in the past 30 years. There are two types of quantum walk, continuous-time quantum walk \cite{childs_spatial_2004} and discrete-time quantum walk \cite{shenvi_quantum_2003} distinguished by the operator of the walk. Quantum walk has been applied to various fields and has become a powerful tool in quantum algorithm design. Quantum walk can be used to solve various problems, such as element distinctness problem \cite{ambainis2007quantum}, triangle finding problem \cite{magniez2007quantum}, and so on. Quantum walk algorithms have been implemented in a variety of physical systems including trapped ions \cite{zahringer2010realization}, trapped atoms \cite{karski2009quantum}, photonic chip \cite{tang2018experimental}, superconducting processor \cite{gong2021quantum}.

Quantum walk search algorithms are quantum algorithms that can search marked vertices in a graph, first introduced by Shenvi Neil, Kempe Julia, and Whaley K. Birgitta in 2003 which studied discrete quantum walk search on hypercube \cite{shenvi_quantum_2003}. Then the quantum walk search algorithm has been studied on various types of graphs, such as complete bipartite graph \cite{rhodes_quantum_2019, peng2024deterministic}, Johnson graph \cite{wong_quantum_2016, peng2024lackadaisical}, tree graph \cite{dimcovic2011framework}, and so on. On some graphs, the quantum walk search algorithm can achieve a quadratic speedup compared to the classical search algorithm. Grover's algorithm can be viewed as a quantum search algorithm on a complete graph with self-loops \cite{ambainis2004coins}, which can achieve a quadratic speedup. For the continuous model of quantum walk, it has been proven that it can search any number of marked vertices on all types of graphs \cite{apers_quadratic_2022}. The Markov model can also provide a universal solution to any search problem \cite{ambainis2020quadratic}.  However, there are still numerous defects in the discrete model of quantum walk, especially coined quantum walk. For instance, the applicability of the model depends on the type of graph and the number of marked vertices. Therefore, we focus on the coined quantum walk in this paper.

Quantum walks on complete multipartite graphs have been studied as follows. In 2009, Reitzner et al. delved into the quantum walk search on a complete multipartite graph with a single marked vertex, achieving quadratic speedup \cite{reitzner_quantum_2009}. In 2018, the investigation expanded to the continuous-time quantum walk search on a complete multipartite graph with one marked vertex, demonstrating the preservation of its quadratic speedup \cite{chiang_optimizing_2018}. Furthermore, in 2021, research on quantum walk search with one marked vertex and state transfer on a multipartite graph with self-loops was studied by Skoupy \cite{skoupy_quantum_2021}. However, there has been no prior research on the search for a complete multipartite graph with multiple marked vertices, which is exactly what this paper will focus on.

As mentioned by Brassard \cite{brassard_searching_1997}, the Grover search algorithm is suffering from a problem that the success probability of finding the marked vertices will oscillate with the increase of steps and the oscillation frequency is dependent on the number of marked vertices. This problem is called the overcooking problem. This results in the inability to search without knowing the number of marked vertices. The problem also exists in quantum walk search. 

Two methods are often applied to overcome the overcooking problem, one of which is using quantum counting to estimate the number of marked vertices before the search process \cite{brassard2002quantum}, and the other is improving the search operators to gain robustness. Grover proposed a fixed-point search algorithm \cite{grover_fixed-point_2005} in 2005 which can search marked vertices in a complete graph without knowing the number of marked vertices, but this method doesn't realize the quadratic speedup. Yoder et al. proposed a fixed-point version of the Grover search \cite{yoder_fixed-point_2014} in 2014 by adding parameters to the Grover search algorithm. The parameterized Grover search algorithm can achieve a quadratic speedup compared to the classical search algorithm, and the success probability of finding the marked vertices will oscillate within a small range close to 1. Xu et al. proposed a robust version of the quantum walk search algorithm on complete bipartite graph in 2022 \cite{xu_robust_2022}, which can also achieve a quadratic speedup. However, robust quantum walk algorithms on other graphs haven't been studied. Thus, we consider the robustness of the quantum walk search algorithm on complete multipartite graphs.

In this paper, we study the quantum walk search algorithm on the complete multipartite graph with multiple marked vertices and consider the robustness of this algorithm. We first introduce the definition of the complete multipartite graph and discrete-time quantum walk search algorithm in section \ref{sec:Preliminaries}, then we apply the quantum walk model to complete multipartite graph with multiple marked vertices in section \ref{sec:nonrobust}. Furthermore, we research the robust version of quantum walk search of the two cases in section \ref{sec:robust}. In section \ref{sec:simulation}, we conduct numerical simulations for further verification of our conclusion. In section \ref{sec: circuit}, we provide the circuit implementation of our algorithm. Finally, we end with a summary in section \ref{sec:conclusion}.

\section{Preliminaries}\label{sec:Preliminaries}
\subsection{Complete M-partite Graph}
A graph G = (V, E) is called M-partite if its vertex set can be partitioned into M disjoint sets such that any two vertices within the same set are not connected. A complete M-partite graph is an M-partite graph in which every pair of vertices from different sets are adjacent. 

In this paper, we will consider two special cases of complete M-partite graph. In both cases, each set contains an equal number of vertices, denoted as $N$. The first case is that the number of the marked vertices in each set is the same, an example is shown in Figure~\ref{fig:fig1}. Each set comprises $n$ marked vertices and ${w}$ unmarked vertices, with a total of $n + {w} = N$ vertices in each set.
The second case is that the marked vertices only exist in one set as shown in Figure~\ref{fig:fig2}. Here, $n$ represents the number of marked vertices in the first set. Without loss of generality, we assume that the marked vertices of the second case are confined to the first set.

In the subsequent sections of the paper, we denote by the symbol $v_{i,j}$ the $j$-th vertex in the $i$-th subset, and without loss of generality, we assume that if there are marked vertices in a subset, the first $n$ vertices are the marked ones.

\begin{figure}[htbp]
    \centering
    \begin{subfigure}{0.4\textwidth}
      \includegraphics[width=\linewidth]{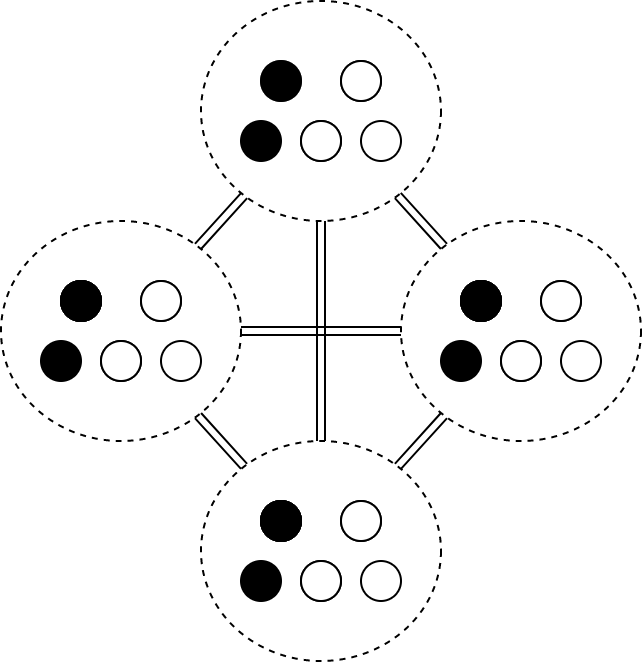}
      \caption{case 1: every set contains marked
vertices}
      \label{fig:fig1}
    \end{subfigure}
    \hfill
    \begin{subfigure}{0.4\textwidth}
      \includegraphics[width=\linewidth]{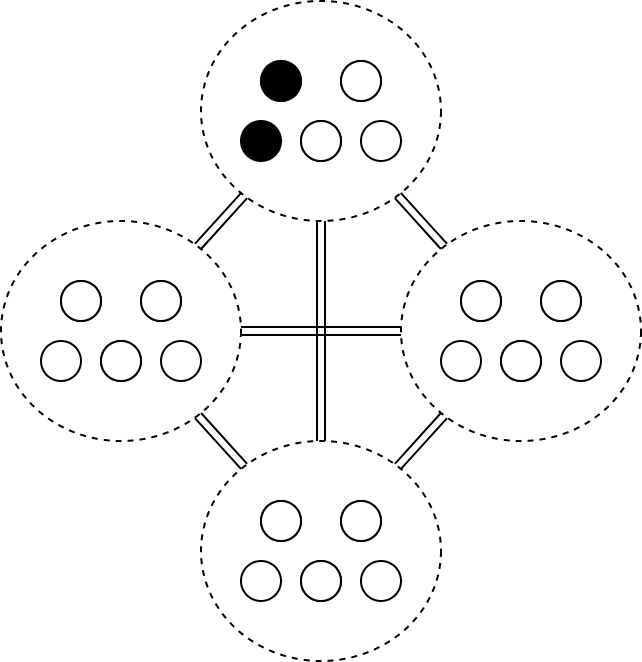}
      \caption{case 2: one set contains marked
vertices}
      \label{fig:fig2}
    \end{subfigure}
    \caption{An example of a complete 4-partite graph of case 1 and case 2 with $N = 5$ and $n = 2$, the black points represent the marked vertices and the white ones represent the unmarked vertices. There are no edges between pairs of vertices within the same set, but there are edges between pairs of vertices from different sets. The black double line means that any vertex in the set at one end of the black double line has an edge to any vertex in the set at the other end.}
    \label{fig:twosubfigures}
\end{figure}

\subsection{Coined Quantum Walk}

The Hilbert space associated with the coined quantum walk operator can be described as $\mathcal{H} = \mathcal H_P \otimes \mathcal H_C  = span\{ \ket{uv}: u,v\in V, (u,v)\in E\}$, where $\ket{uv}$ is the state that the walker is at vertex $u$ and the coin is in state $v$, $\mathcal H_P$ is the position space, and the coin space $\mathcal{H}_C$ represents the direction indicating the next walk step to take. The walk operator is given by $W = SC$, where $S$ is the shift operator defined as
\begin{align*}
    S\ket{uv} = \ket{vu},
\end{align*}
and $C$ is the coin operator defined as 
\begin{align*}
    C = \sum_u\ket{u}\bra{u}\otimes C_u,
\end{align*}
$C_u$ is the Grover diffusion operator often defined as $C_u = 2\ket{s_u}\bra{s_u} - I$, $\ket{s_u} = \frac{1}{\sqrt{d_u}} \sum_{v\in N(u)}\ket{v}$, in which $d_u$ is the degree of $u$ in graph and $N(u)$ is the set of neighbours of $u$.
The coin operator and shift operator are both unitary operators acting on the Hilbert space $\mathcal{H}$.

To find marked vertices in a graph, we use a query oracle $Q$ defined by 
\begin{align*}
    Q\ket{uv} = \begin{cases}
        -\ket{uv} &  u~is~marked,\\
        \ket{uv} &  u~is~not~marked.
    \end{cases}
\end{align*}
The search operator is given by $U = WQ = SCQ$ which is also a unitary operator acting on the Hilbert space $\mathcal{H}$.
The initial state is often chosen as the uniform superposition state
\begin{align}
    \ket{\psi_0 } = \frac 1 {\sqrt{2 d}} \sum_{u \in V} \sum_{v \in N(u)} \ket{uv},
\end{align}
where $d$ is the total number of edges in the graph.
The quantum state after $t$ steps is 
\begin{align}
    \ket{\psi_t} = U^{t}\ket{\psi_0},
\end{align}
and then we measure the first and second registers until finding a marked vertex, the success probability is 
\begin{align}
    p_{succ}(t) = \sum_{u\ or\ v\ is\ marked}|\langle{\psi_t}|{uv}\rangle|^2.
\end{align}

In this paper, the ``$\sim$" signifies proportionality, meaning the expressions on both sides of it differ by a constant factor, and the ``$\simeq$" signifies asymptotically equal, meaning the expressions on both sides of it are approximately equal in the limit sense.

\section{Quantum walk search on Complete Multipartite Graph}\label{sec:nonrobust}
\subsection{Case 1: Every Set Contains Marked Vertices}\label{sec:nonrobust_case1}
 
We first consider the case that the number of vertices and marked vertices in each set is the same. Vertices can be classified into two types, marked vertices denoted by $a$ and unmarked vertices denoted by $b$.
Therefore the search operator $U$ in the Hilbert space $\mathcal{H}$ can be described as an operator on a four dimension invariant subspace $span\{ \ket{aa}, \ket{ab}, \ket{ba}, \ket{bb}\}$, which can be shown in Figure~\ref{fig:fig3}. These four basis states $\ket{aa}, \ket{ab}, \ket{ba}, \ket{bb}$ are given below:
\begin{figure}[htbp]
    \centering
    \includegraphics[width=0.45\textwidth]{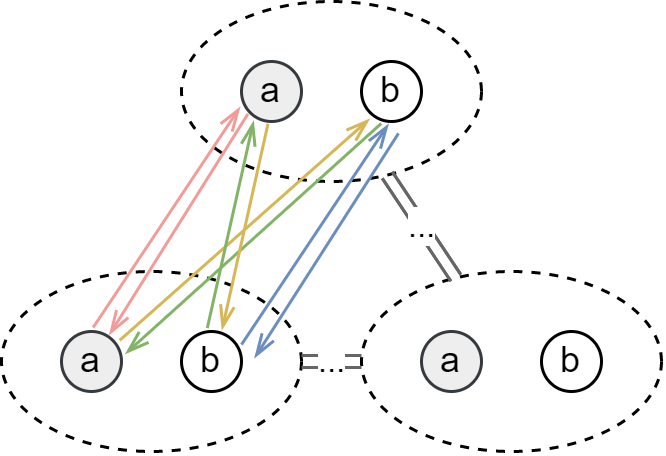}
    \caption{Complete M-partite graph of case 1, grey vertices are marked vertices. Red arrows denote states in $\ket{aa}$, yellow in $\ket{ab}$, green in $\ket{ba}$, and blue in $\ket{bb}$.
    The relationships between any other two pairs of sets are also as described earlier.
    }
    \label{fig:fig3}
\end{figure}

\begin{align*}
|aa\rangle  &= \frac 1 {\sqrt{Mn^2(M-1)}} \sum_{i_1 = 1}^M\sum_{j_1 = 1}^n \sum_{\begin{subarray}
    \\ i_2 = 1 \\ i_2\ne i_1
\end{subarray}}^M \sum_{j_2 = 1}^n |v_{i_1,j_1}v_{i_2,j_2}\rangle,\\
|ab\rangle &= \frac 1 {\sqrt{Mn(M-1){w}}} \sum_{i_1 = 1}^M\sum_{j_1 = 1}^n \sum_{\begin{subarray}
    \\ i_2 = 1 \\ i_2\ne i_1
\end{subarray}}^M \sum_{j_2 = n+1}^{n+{w}} |v_{i_1,j_1}v_{i_2,j_2}\rangle,\\
|ba\rangle &= \frac 1 {\sqrt{Mn(M-1){w}}} \sum_{i_1 = 1}^M\sum_{j_1 = n+1}^{n+{w}} \sum_{\begin{subarray}
    \\ i_2 = 1 \\ i_2\ne i_1
\end{subarray}}^M \sum_{j_2 = 1}^n |v_{i_1,j_1}v_{i_2,j_2}\rangle,\\
|bb\rangle &= \frac 1 {\sqrt{M{{w}}^2(M-1)}} \sum_{i_1 = 1}^M\sum_{j_1 = n+1}^{n+{w}} \sum_{\begin{subarray}
    \\ i_2 = 1 \\ i_2\ne i_1
\end{subarray}}^M \sum_{j_2 = n+1}^{n+{w}} |v_{i_1,j_1}v_{i_2,j_2}\rangle.
\end{align*}
The initial state can be represented as 
\begin{align*}
\ket {\psi_0} &= \frac 1 {\sqrt{M(M-1)N^2}} \sum_{i_1 = 1}^M\sum_{j_1 = 1}^{N} \sum_{\begin{subarray}\\ i_2 = 1 \\ i_2\ne i_1\end{subarray}}^M \sum_{j_2 = 1}^{N} |v_{i_1,j_1}v_{i_2,j_2}\rangle\\
    &= \frac{n}{N} \ket{aa} + \frac{{w}}{N} \ket{bb} + \frac{\sqrt{n {w}}}{N}(\ket{ab} + \ket{ba})\\ 
    &= \frac 1 N \left( \begin{matrix}
        n\\ \sqrt{n {w}}\\ \sqrt{n {w}}\\ {w}
    \end{matrix} \right),
\end{align*}
and the search operators under this basis are given by 
\begin{align*}
S &= \left(
\begin{array}{cccc}
 1 & 0 & 0 & 0 \\
 0 & 0 & 1 & 0 \\
 0 & 1 & 0 & 0 \\
 0 & 0 & 0 & 1 \\
\end{array}
\right), \\
C &= I \otimes \left(\begin{matrix}
    2 \frac n N - 1 & 2 \frac {\sqrt{n {w}}} N\\
    2 \frac {\sqrt{n {w}}} N & 2 \frac {{w}} N - 1
\end{matrix}\right),
\\
Q &= \left(
\begin{array}{cccc}
 -1 & 0 & 0 & 0 \\
 0 & -1 & 0 & 0 \\
 0 & 0 & 1 & 0 \\
 0 & 0 & 0 & 1 \\
\end{array}
\right),\\
U &= \left(
\begin{array}{cccc}
 1-\frac{2 n}{n+{w}} & -\frac{2 \sqrt{n {w}}}{n+{w}} & 0 & 0 \\
 0 & 0 & \frac{2 n}{n+{w}}-1 & \frac{2 \sqrt{n {w}}}{n+{w}} \\
 -\frac{2 \sqrt{n {w}}}{n+{w}} & 1-\frac{2 {w}}{n+{w}} & 0 & 0 \\
 0 & 0 & \frac{2 \sqrt{n {w}}}{n+{w}} & \frac{2 {w}}{n+{w}}-1\\
\end{array}
\right).
\end{align*}


We can calculate the probability of finding the walker at marked vertices after $t$ steps by computing the eigenvalues and eigenstates of $U$. The eigenvalues in the limit of large-M and large-N are $-1,1,\e^{-\ii\omega}, \e^{\ii\omega}$, in which $\omega$ is given by $\omega = \arccos(\frac {{w} -n}{n +{w}})$. 
The corresponding eigenstates are
\begin{equation}\begin{aligned}
-1:& \ket{v_1} = \sqrt{\frac{n}{2(n+{w})}}\left(\begin{matrix}-1 & -\sqrt{\frac {{w}} n} &  -\sqrt{\frac {{w}} n} & 1\end{matrix}\right)^\top ,\\
1:& \ket{v_2} = \sqrt{\frac{{w}}{2(n+{w})}}\left(\begin{matrix}-1 & \sqrt{\frac n {{w}} }&  \sqrt{\frac n {{w}} } &1 \end{matrix}\right)^\top ,\\
\e^{-\ii\omega}:& \ket{v_3} = \left(\begin{matrix}\frac{1}{2} & \frac{\ii}{2} & -\frac{\ii}{2} & \frac{1}{2}\end{matrix}\right)^\top ,\\
\e^{\ii\omega}:& \ket{v_4} = \left(\begin{matrix}\frac{1}{2} & -\frac{\ii}{2} & \frac{\ii}{2} & \frac{1}{2}\end{matrix}\right)^\top ,
\end{aligned}\label{eq:3_1_eigen}\end{equation}
where the superscript $\top$ indicates transpose.

Then we can decompose the initial state $\ket{\psi_0}$ into the eigenstates of $U$ when $N\gg n$:
\begin{align*}
    \ket{\psi_0} \simeq& (0,0,0,1)^\top  \\
    =& \sqrt{\frac{n}{2(n +{w})}} \ket{v_1} + \sqrt{\frac{{w}}{2(n +{w})}} \ket{v_2} \\&+ \frac 1 2 \ket{v_3} + \frac 1 2\ket{v_4}.
\end{align*}
Thus the state after $t$ steps is 
\begin{equation}\begin{aligned}
    |\psi_t\rangle 
    =& U^t |\psi_0\rangle \\
    \simeq& \sqrt\frac{n}{2(n +{w})} (-1) ^t |v_1\rangle + \sqrt\frac{{w}}{2(n +{w})} 1 ^t |v_2\rangle \\
    &+ \frac 1 2 \e^{-\ii\omega t} |v_3\rangle + \frac 1 2 \e^{\ii\omega t} |v_4\rangle.
    \label{eq:3_1_psi_t}
\end{aligned}\end{equation}

The probability of finding the walker at marked vertices after $t$ steps can be formalized as
\begin{align*}
    p_{succ}(t) = 1 - \langle \psi_t |bb\rangle \langle bb| \psi_t \rangle = 1 - |\langle \psi_t | bb\rangle| ^2.
\end{align*}
Combining Eq.\eqref{eq:3_1_eigen} and Eq.\eqref{eq:3_1_psi_t}, we can obtain  
\begin{align*}
    \langle \psi_t | bb\rangle =
    \begin{cases}
        \frac 12 \cos(\omega t) + \frac 12 & t \text{ is even},\\
        \frac 12 \cos(\omega t) + \frac {{w} - n}{2 (n +{w})}  & t \text{ is odd}.
    \end{cases}
\end{align*}

When t is an even number, the maximum success probability is achieved at $\omega t = \pi$, and when t is an odd number, the maximum success probability is attained at $\omega t = \arccos( \frac {n - {w}}{n +{w}} )$. The maximum probability of finding the walker at marked vertices is $p_{succ}(t) = 1$. So we can choose the number of steps $t$ as the nearest even number to $\frac{\pi}{\omega}$ or the nearest odd number to $\frac{\arccos( \frac {n - {w}}{n +{w}} )}{\omega}$ to achieve the maximum probability of finding the walker at marked vertices.

Assuming that $N \gg n$, the step to achieve the maximum probability is

\begin{align}
     t= \begin{cases}
          \frac{\pi}{\arccos(\frac{{w} -n}{n+{w}})} \simeq \frac {\pi}{2} \sqrt{\frac{N}{n}} = O(\sqrt{\frac{N}{n}}) & t \text{ is even}, \\
          \frac{\arccos( \frac {n - {w}}{n +{w}} )}{\arccos(\frac{{w} -n}{n+{w}})}
          \simeq \frac {\pi}{2} \sqrt{\frac{N}{n}} = O(\sqrt{\frac{N}{n}}) & t \text{ is odd}.
     \end{cases}
\end{align}
Because there are $MN$ vertices and $Mn$ marked vertices in total, the complexity of the classical search algorithm is $O(\frac{MN}{Mn}) = O(\frac Nn)$, thus the quantum algorithm can achieve a quadratic speedup compared to the classical search algorithm.

\subsection{Case 2: One Set Contains Marked Vertices}\label{sec:nonrobust_case2}

In this subsection, we consider the case that the marked vertices only exist in the first set noted as set 1. 
Vertices can be classified into 3 types, marked vertices in set 1 denoted by $a$, unmarked vertices in set 1 denoted by $b$, and unmarked vertices in other sets denoted by c.

Therefore the search operator on the Hilbert space can be described as an operator on a five dimension subspace $span\{\ket{ca}, \ket{ac}, \ket{bc}, \ket{cb}, \ket{cc}\}$, where $\{ \ket{ca}, \ket{ac}, \ket{bc}, \ket{cb}, \ket{cc}\}$ are shown in Figure~\ref{fig:fig4} and can be given by the following formulas:  

\begin{figure}[htbp]
    \centering
    \includegraphics[width=0.45\textwidth]{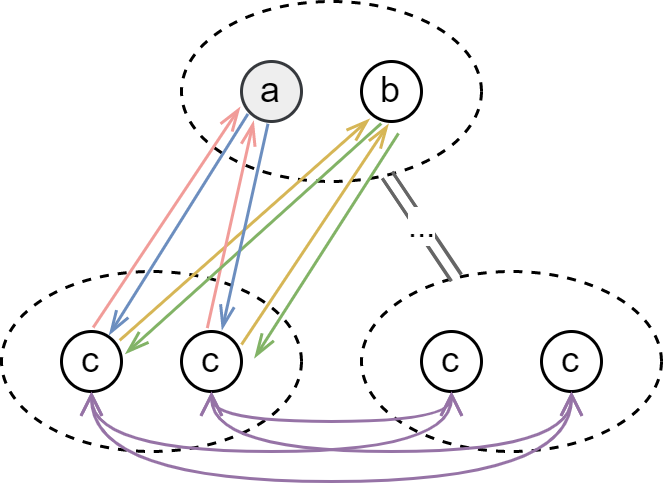}
    \caption{Complete 3-partite graph of case 2, grey vertices are marked vertices.
    Red arrows represent states in $\ket{ca}$, blue in $\ket{ac}$, green in $\ket{bc}$, yellow in $\ket{cb}$, and purple in $\ket{cc}$.}
    \label{fig:fig4}
\end{figure}

\begin{align*}
    \ket{ca} &= \frac{\sum _{m=2}^M \sum _{n_1=1}^N \sum _{n_2=1}^n |v_{m,n_1} v_{1,n_2}\rangle }{\sqrt{(M-1) n N}},\\ 
    \ket{ac} &= \frac{\sum _{m=2}^M \sum _{n_1=1}^N \sum _{n_2=1}^n |v_{1,n_2} v_{m,n_1}\rangle }{\sqrt{(M-1) n N}},\\ 
    \ket{bc} &= \frac{\sum _{m=2}^M \sum _{n_1=n+1}^N \sum _{n_2=1}^N |v_{1,n_1} v_{m,n_2}\rangle }{\sqrt{(M-1) N (N-n)}},\\
    \ket{cb} &= \frac{\sum _{m=2}^M \sum _{n_1=n+1}^N \sum _{n_2=1}^N |v_{m,n_2} v_{1,n_1}\rangle }{\sqrt{(M-1) N (N-n)}},\\
    \ket{cc} &= \frac{\sum _{m_1=2}^M \sum _{
       \begin{subarray} 
       \\m_2 = 2 \\ m_2 \ne m_1 
       \end{subarray}}
    ^M \sum _{n_1=1}^N \sum _{n_2=1}^N |v_{m_1,n_1} v_{m_2,n_2}\rangle }{\sqrt{(M-2) (M-1) N^2}}.
\end{align*}
When $M$ is large enough, the initial state is approximate to
\begin{align*}
    \ket{\psi_0} \simeq \ket{cc} = \left( \begin{matrix}
        0&0&0&0&1
    \end{matrix} \right)^\top ,
\end{align*}
and the search operators under this basis are given by
\begin{align*}
    S = \left(
\begin{array}{ccccc}
 0 & 1 & 0 & 0 & 0 \\
 1 & 0 & 0 & 0 & 0 \\
 0 & 0 & 0 & 1 & 0 \\
 0 & 0 & 1 & 0 & 0 \\
 0 & 0 & 0 & 0 & 1 \\
\end{array}
\right), \ Q= \left(
\begin{array}{ccccc}
 1 & 0 & 0 & 0 & 0 \\
 0 & -1 & 0 & 0 & 0 \\
 0 & 0 & 1 & 0 & 0 \\
 0 & 0 & 0 & 1 & 0 \\
 0 & 0 & 0 & 0 & 1 \\
\end{array}
\right),
\end{align*}

\begin{widetext}
\begin{align*}
C = \left(
\begin{array}{ccccc}
 \frac{2 n}{(M-1) N}-1 & 0 & 0 & \frac{2 \sqrt{n (N-n)}}{(M-1) N} & \frac{2 \sqrt{\frac{(M-2) n}{N}}}{M-1} \\
 0 & 1 & 0 & 0 & 0 \\
 0 & 0 & 1 & 0 & 0 \\
 \frac{2 \sqrt{n (N-n)}}{(M-1) N} & 0 & 0 & \frac{2 (n-N)}{N-M N}-1 & \frac{2 \sqrt{\frac{(M-2) (N-n)}{N}}}{M-1} \\
 \frac{2 \sqrt{\frac{(M-2) n}{N}}}{M-1} & 0 & 0 & \frac{2 \sqrt{\frac{(M-2) (N-n)}{N}}}{M-1} & \frac{2 (M-2)}{M-1}-1 \\
\end{array}
\right),
U = \left(
\begin{array}{ccccc}
 0 & -1 & 0 & 0 & 0 \\
 \frac{2 n}{(M-1) N}-1 & 0 & 0 & \frac{2 \sqrt{n (N-n)}}{(M-1) N} & \frac{2 \sqrt{\frac{(M-2) n}{N}}}{M-1} \\
 \frac{2 \sqrt{n (N-n)}}{(M-1) N} & 0 & 0 & \frac{2 (n-N)}{N-M N}-1 & \frac{2 \sqrt{\frac{(M-2) (N-n)}{N}}}{M-1} \\
 0 & 0 & 1 & 0 & 0 \\
 \frac{2 \sqrt{\frac{(M-2) n}{N}}}{M-1} & 0 & 0 & \frac{2 \sqrt{\frac{(M-2) (N-n)}{N}}}{M-1} & \frac{2 (M-2)}{M-1}-1 \\
\end{array}
\right).
\end{align*}
\end{widetext}

Using the similar method shown in case 1, we can calculate the probability of finding the walker at marked vertices after $t$ steps by computing the spectral decomposition of search operator $U$. 

The eigenvalues of $U$ are $-1,\frac{N - n}{N(M-1)} \pm \ii, 1 \pm \sqrt{\frac {2 n}{MN}} \ii$ in the limit of large-M and large-N, and the corresponding eigenstates are 
\begin{align*}
    -1:& \ket{v'_1} = (\frac 1 {\sqrt{2}},\frac 1 {\sqrt{2}},0,0,0)^\top ,\\
    \frac{n - N}{N(M-1)} + \ii:& \ket{v'_2} = (0,0,\frac 1 {\sqrt{2}},- \frac \ii {\sqrt{2}},0)^\top ,\\
    \frac{n - N}{N(M-1)} - \ii:& \ket{v'_3} = (0,0,\frac 1 {\sqrt{2}}, \frac \ii {\sqrt{2}},0)^\top ,\\
    1 + \sqrt{\frac {2 n}{MN}} \ii:& \ket{v'_4} = ( \frac \ii2 ,- \frac \ii 2,0,0,\frac 1 {\sqrt{2}})^\top ,\\
    1 - \sqrt{\frac {2 n}{MN}} \ii:& \ket{v'_5} = (- \frac \ii2 ,\frac \ii 2,0,0,\frac 1 {\sqrt{2}})^\top .
\end{align*}
Thus the initial state can be decomposed as
\begin{align*}
    \ket{\psi_0} \simeq |cc\rangle = \frac{1}{\sqrt 2} (\ket{v'_4} + \ket{v'_5}).
\end{align*}
So the state after $t$ steps is
\begin{align*}
\ket{\psi_t} = U^t \ket{\psi_0} 
\simeq \frac 1{\sqrt{2}}(\e^{i{\omega'} t} \ket{v'_4} +  \e^{-\ii{\omega'} t}\ket{v'_5})\\
= \frac 1{\sqrt{2}} \left( \begin{matrix}
    -\sin({\omega'} t)\\ \sin({\omega'} t)\\0\\0\\\sqrt{2}\cos({\omega'} t)
\end{matrix}\right),
\end{align*}
in which ${\omega'}$ is given by ${\omega'} = \arcsin(\sqrt{\frac {2 n}{MN}})$. 


The success probability of finding the walker at marked vertices after $t$ steps is given by
\begin{align*}
    p_{succ}(t) &= 1 - |\langle \psi_t | cc\rangle |^2 
 - |\langle \psi_t | bc \rangle |^2 - |\langle \psi_t | cb\rangle |^2\\ 
 &= 1 - \cos({\omega'} t)^2.
\end{align*}
When ${\omega'} t= \frac{\pi}{2}$, the probability of finding the walker at marked vertices after $t$ steps reaches the maximum.  The probability of finding the walker at marked vertices is $p_{succ}(t) = 1$. So we can choose the number of steps $t$ as the nearest number to $\frac{\pi}{2 \omega'}$ to achieve the maximum probability of finding the walker at marked vertices. Assuming $MN \gg n$, $t \sim \frac 1 {\omega'} = \frac{1}{\arcsin(\sqrt{\frac {2 n}{MN}})}\simeq \sqrt{\frac {MN}{2n}}$, thus the complexity is $O(\sqrt{\frac {MN}{n}})$, quantum walk search for this case retains quadratic speedup.

\section{Robust Quantum Walk On Complete Multipartite Graph}\label{sec:robust}

\subsection{Robust Quantum Walk Algorithm}

The method of adding parameters to the quantum operator is typically used to improve quantum algorithms, such as constructing deterministic algorithms\cite{li2023deterministic} or robust algorithms\cite{yoder_fixed-point_2014}.
Here, we fix the search algorithm mentioned before by adding parameters to query oracle $Q$ and the coin operator $C$. The new operators of the robust quantum search operator are
\begin{align*}
    U(\alpha, \beta) = SC(\alpha)Q(\beta),
\end{align*}
in which the coin operator $C$ is changed to
\begin{align*} C(\alpha)=\sum_u\left|u\right>\left<u\right|\otimes\left[\left(1-\e^{-\ii\alpha}\right)\left|s_u\right>\left<s_u\right|-I\right],
\end{align*}
and the query oracle $Q$ is changed to 
\begin{align*}
    Q(\beta)|uv\rangle=\begin{cases}
    {\e^{\ii\beta}|uv\rangle}&\mathrm{~if~}u\mathrm{~is~marked},\\
    |uv\rangle&\mathrm{~if~}u\mathrm{~is~not~marked}.
    \end{cases}
\end{align*}
Specifically, when $\alpha = \beta = \pm \pi$, the parameterized operators become the original operators.

\subsection{Case 1: Every Set Contains Marked Vertices}\label{sec:robust_case1}

The parameterized operators in subspace $span\{ \ket{aa}, \ket{ab}, \ket{ba}, \ket{bb}\}$ in matrix form are
\begin{equation}\begin{aligned}\label{eq:robustC}
&C(\alpha) = I \otimes \left(\begin{matrix}
    (1-\e^{-\ii\alpha}) \frac n N - 1 & (1-\e^{-\ii\alpha})  \frac {\sqrt{n {w}}} N\\
    (1-\e^{-\ii\alpha})  \frac {\sqrt{n {w}}} N & (1-\e^{-\ii\alpha}) \frac {w} N - 1
\end{matrix}\right) \\
&=  I \otimes \left(\begin{matrix}
    (1-\e^{-\ii\alpha}) \frac {1- \cos \omega} 2 - 1 & (1-\e^{-\ii\alpha})  \frac {\sin\omega} 2\\
    (1-\e^{-\ii\alpha})  \frac {\sin\omega} 2 & (1-\e^{-\ii\alpha})\frac {1+ \cos \omega} 2 - 1
\end{matrix}\right),
\end{aligned}\end{equation}

\begin{equation}\begin{aligned}\label{eq:robustQ}
Q(\beta) &= \left(
\begin{array}{cccc}
 \e^{\ii\beta} & 0 & 0 & 0 \\
 0 & \e^{\ii\beta} & 0 & 0 \\
 0 & 0 & 1 & 0 \\
 0 & 0 & 0 & 1 \\
\end{array}
\right),
\end{aligned}\end{equation}
where $\omega= \arccos\frac {{w} -n}{n+{w}}$.

\begin{theorem}\label{th:theorem1}
Consider a complete multipartite graph comprising $M$ subsets, each with $N$ vertices, where $n$ vertices in each subset are marked, and the value of $n$ is unknown.
Let $\epsilon$ be a number within the range $\left(0,1\right]$, and $t$ be an integer satisfying $t \ge \ln\left(\frac{2}{\sqrt{\epsilon}}\right) \sqrt{N} + 1$, there exist a series of parameters  $\alpha_1, \ldots, \alpha_t$ and $\beta_1, \ldots, \beta_{t+1}$, so that the probability of obtaining a marked vertex exceeds $1-\epsilon^2$ after the sequential application of unitary operations $U(\alpha_1, \beta_2)U(\alpha_1, \beta_1), U(\alpha_2, \beta_3)(\alpha_2, \beta_2), \ldots,$ $U(\alpha_t, \beta_{t+1})U(\alpha_t, \beta_t)$ to the uniform superposition initial state. 
\end{theorem}

Then we give a proof of Theorem~\ref{th:theorem1}. We can transform the four-dimensional operator $U(\alpha, \beta)$ into a two-dimensional operator by using the Kronecker product of two-dimensional matrices. Then it becomes easy to analyze the success probability of the quantum walk algorithm. 

\begin{proof}[Proof of Theorem~\ref{th:theorem1}]
Consider Eq.\eqref{eq:robustC} and Eq.\eqref{eq:robustQ} we can find that the parameterized coin operator $C(\alpha)$ only acts on the coin register, and the parameterized query operator $Q(\beta)$ only acts on the position register. Therefore, both of them can be written as the Kronecker product of a two-dimensional matrix and an identity matrix.

The parameterized coin operator $C(\alpha)$ can be transformed into
\begin{equation}\begin{aligned}
    C(\alpha) &= I \otimes \left(\begin{matrix}
        (1-\e^{-\ii\alpha}) \frac n N - 1 & (1-\e^{-\ii\alpha})  \frac {\sqrt{n {w}}} N\\
        (1-\e^{-\ii\alpha})  \frac {\sqrt{n {w}}} N & (1-\e^{-\ii\alpha}) \frac {w} N - 1
    \end{matrix}\right) \\
    &= I \otimes C_{sub}(\alpha),
\end{aligned}\end{equation}
where $w$ is the number of marked vertices in each subset, and
\begin{align*}
    C_{sub}(\alpha) = \left(\begin{matrix}
        (1-\e^{-\ii\alpha}) \frac n N - 1 & (1-\e^{-\ii\alpha})  \frac {\sqrt{n {w}}} N\\
        (1-\e^{-\ii\alpha})  \frac {\sqrt{n {w}}} N & (1-\e^{-\ii\alpha}) \frac w N - 1
    \end{matrix}\right).
\end{align*}The parameterized query operator $Q(\beta)$ can be transformed into
\begin{equation}\begin{aligned}
    Q(\beta) &= \left(
        \begin{array}{cccc}
        \e^{\ii\beta} & 0 & 0 & 0 \\
        0 & \e^{\ii\beta} & 0 & 0 \\
        0 & 0 & 1 & 0 \\
        0 & 0 & 0 & 1 \\
        \end{array}
        \right) =\left(
        \begin{array}{cc}
        \e^{\ii\beta} & 0 \\
        0 & 1  \\
        \end{array}
        \right)  \otimes I = Q_{sub}(\beta) \otimes I,
\end{aligned}\end{equation}
where
\begin{align*}
    Q_{sub}(\beta) &= \left(
        \begin{array}{cc}
        \e^{\ii\beta} & 0 \\
        0 & 1  \\
        \end{array}
        \right).
\end{align*}

Then we can transform the four-dimensional operator $U(\alpha, \beta)$ into a two-dimensional operator by using the Kronecker product of two-dimensional matrices. After simple calculations, we can get
\begin{align*}
    C(\alpha) Q(\beta) = 
    Q_{sub}(\beta)\otimes C_{sub}(\alpha),
\end{align*}
\begin{align*} 
    &S C(\alpha) Q(\beta) S \\
    =& (S C(\alpha) S)\cdot(S Q(\beta) S) \\ 
    =& (C_{sub}(\alpha)\otimes I)\cdot(I \otimes Q_{sub}(\beta)) \\ 
    =& C_{sub}(\alpha)\otimes Q_{sub}(\beta).
\end{align*}

By using the above formula, we can analyze the state after $2t$ steps of the robust quantum walk algorithm by splitting the effect of quantum walk operators on the position register and coin register. The state $|\psi_{2t}\rangle$ after $2t$ steps can be written as the tensor product of two two-dimensional vectors:

\begin{widetext}
    \begin{align*}
        |\psi_{2t}\rangle
        =& (U(\alpha_t, \beta_{t+1}) U(\alpha_t, \beta_t)) \cdot ...\cdot (U(\alpha_1, \beta_2) U(\alpha_1, \beta_1)) \ket{\bar 0}\\
        =& \prod_{t\ge i\ge 1}\left( U(\alpha_i, \beta_{i+1})\cdot U(\alpha_i, \beta_{i}) \right)\cdot \ket{\bar 0} \\
        =& \prod_{t\ge i\ge 1} \left( S C(\alpha_i) Q(\beta_{i+1})S C(\alpha_i) Q(\beta_{i}) \right) \ket{\bar 0} \\ 
        =& \prod_{t\ge i\ge 1} ((C_{sub}(\alpha_i)\otimes Q_{sub}(\beta_{i+1})) \cdot (Q_{sub}(\beta_{i})\otimes C_{sub}(\alpha_i)))\ket{\bar 0}\\ 
        =&\left( \prod_{t\ge i\ge 1} (C_{sub}(\alpha_i) Q_{sub}(\beta_i)) \ket{\bar 0}_2 \right)\otimes
        \left(\prod_{t\ge i\ge 1} (Q_{sub}(\beta_{i+1}) C_{sub}(\alpha_i))\ket{\bar 0}_2\right)\\
        =&\left( \prod_{t\ge i\ge 1} (C_{sub}(\alpha_i) Q_{sub}(\beta_i)) \ket{\bar 0}_2 \right)\otimes
        \left(\prod_{t\ge i\ge 1} (Q_{sub}(\beta_{i+1}) C_{sub}(\alpha_i)) Q_{sub}(\beta_1) \ket{\bar 0}_2\right)\\
        =& \left(\prod_{t\ge i\ge 1} (C_{sub}(\alpha_i) Q_{sub}(\beta_i)) \ket{\bar 0}_2 \right)\otimes
        \left( Q_{sub}(\beta_{t+1}) \prod_{t\ge i\ge 1} (C_{sub}(\alpha_i) Q_{sub}(\beta_i)) \ket{\bar 0}_2\right),
    \end{align*}
\end{widetext}
where $\ket{\bar 0}$ is a 4-dimensional vector $(0,0,0,1)^\top$, and $\ket{\bar 0}_2$ is a 2-dimensional vector $(0,1)^\top$, $\prod_{t\ge i \ge 1}$ means multiplying the item sequentially for index from $t$ to $1$, this notation is also used in other parts of this paper.

Next we just need to analyze the property of $\prod_{t\ge i\ge 1} (C_{sub}(\alpha_i) Q_{sub}(\beta_i)) \ket{\bar 0}_2$, where $C_{sub}(\alpha_i)$ and $Q_{sub}(\beta_i)$ are both 2-dimensional matrix. According to Ref.\cite{yoder_fixed-point_2014}, when $t$ satisfies $t \ge \ln\left(\frac{2}{\sqrt{\epsilon}}\right) \sqrt{N} + 1$, for $j = 1,2,...,t$ we can choose parameters $\alpha_j$ and $\beta_j$ satisfying
\begin{align*}
    \alpha_j=-\beta_{t-j+1}=2\mathrm{arccot} \biggl(\tan(2\pi j/L)\sqrt{1-\gamma^2}\biggr),
\end{align*}
where $L = 2t+ 1$ is an integer, and $\gamma = 1/T_{1/L}(1/\sqrt{\epsilon})$, $T_{L}(x)=\mathrm{cos}(L\mathrm{arccos}(x))$ is Chebyshev's polynomial, then we have
\begin{equation}\begin{aligned}\label{eq:robust_success_2x2}
    |\langle \bar 0 |_2 \prod_{t\ge i\ge 1} (C_{sub}(\alpha_i) Q_{sub}(\beta_i)) \ket{\bar 0}_2|^2 \le \epsilon.
\end{aligned}\end{equation}
    
The parameter $\beta_{t+1}$ does not affect the success probability and thus can be chosen arbitrarily.

From Eq.\eqref{eq:robust_success_2x2}, the upper bound of the failure probability $|\langle bb | \psi_{2t}\rangle|^2$ can be further calculated as follows:
\begin{align*}
    |\langle bb | \psi_{2t}\rangle|^2 
    =& \left(|\langle \bar 0 |_2 \prod_{t\ge i\ge 1} (C_{sub}(\alpha_i) Q_{sub}(\beta_i)) \ket{\bar 0}_2|^2\right) \\&\cdot \left(|\langle \bar 0 |_2 Q_{sub}(\beta_0) \prod_{t\ge i\ge 1} (C_{sub}(\alpha_i) Q_{sub}(\beta_i)) \ket{\bar 0}_2|^2\right) \\ 
    \le& \epsilon^2.
\end{align*}

Thus the success probability of the quantum walk search algorithm is
\begin{align*}
    p_{succ}(2t) = 1 - |\langle bb | \psi_{2t}\rangle|^2 \ge 1 - \epsilon^2.
\end{align*}

This establishes the proof of Theorem~\ref{th:theorem1}.
\end{proof}

\subsection{Case 2: One Set Contains Marked Vertices}\label{sec:robust_case2}

In this section, we will give a robust quantum walk algorithm for the case that the marked vertices only exist in the first set. We will simplify the operator matrix using an approximate method to analyze the state after $2t$ steps, thereby completing the proof. 

Theorem~\ref{th:theorem2} gives the main result of this section.

\begin{theorem}\label{th:theorem2}
    Consider a complete multipartite graph comprising $M$ subsets, each with $N$ vertices, where $n$ vertices in the first subset are marked, and the value of $n$ is unknown.
    Let $\epsilon$ be a number within the range $\left(0,1\right]$, and $t$ be an integer satisfying $t \ge \ln\left(\frac{2}{\sqrt{\epsilon}}\right) \sqrt{\frac{MN}2} + 1$, there exist a series of parameters $\alpha_1, \ldots, \alpha_t$ and $\beta_1, \ldots, \beta_{t+1}$, so that the probability of obtaining a marked vertex exceeds $1-\epsilon$ after the sequential application of unitary operations $U(\alpha_1, \beta_2)U(\alpha_1, \beta_1), U(\alpha_2, \beta_3)(\alpha_2, \beta_2), \ldots, U(\alpha_t, \beta_{t+1})U(\alpha_t, \beta_t)$ to the uniform superposition initial state.
\end{theorem}

The parameterized operators $U(\alpha, \beta)$ in the subspace $span\{ \ket{ca}, \ket{ac}, \ket{bc}, \ket{cb}, \ket{cc}\}$ in matrix form is
\begin{widetext}
\begin{tiny}    
\begin{align*}
U(\alpha,\beta) = \left(
\begin{array}{ccccc}
 0 & -\e^{-\ii (\alpha -\beta )} & 0 & 0 & 0 \\
 -1+\frac{\left(1-\e^{-\ii \alpha }\right) n}{(M-1) \text{N}} & 0 & 0 & \frac{\left(1+\e^{-\ii \alpha }\right) \sqrt{n (\text{N}-n)}}{(M-1) \text{N}} & \frac{\left(1+\e^{-\ii \alpha }\right) \sqrt{(M-2) n \text{N}}}{(M-1) \text{N}} \\
 \frac{\left(1+\e^{-\ii \alpha }\right) \sqrt{n (\text{N}-n)}}{(M-1) \text{N}} & 0 & 0 & -1+\frac{\left(1-\e^{-\ii \alpha }\right) (\text{N}-n)}{(M-1) \text{N}} & \frac{\left(1+\e^{-\ii \alpha }\right) \sqrt{\frac{(M-2) (\text{N}-n)}{\text{N}}}}{M-1} \\
 0 & 0 & -\e^{-\ii \alpha } & 0 & 0 \\
 \frac{\left(1+\e^{-\ii \alpha }\right) \sqrt{(M-2) n \text{N}}}{(M-1) \text{N}} & 0 & 0 & \frac{\left(1+\e^{-\ii \alpha }\right) \sqrt{\frac{(M-2) (\text{N}-n)}{\text{N}}}}{M-1} & -1+\frac{\left(1-\e^{-\ii \alpha }\right) (M-2)}{M-1} x\\
\end{array}
\right).
\end{align*}
\end{tiny}
\end{widetext}

First, we give a lemma which is useful for the proof of Theorem~\ref{th:theorem2}.

\begin{lemma}\label{lemma:lemma1}
    When $M,N$ are large enough and $k$ is any positive integer, $|\psi_{2k}\rangle = \prod_{k\ge i\ge 1}U(\alpha_i,\beta_{i+1}) \cdot U(\alpha_i, \beta_i) |\bar{0}\rangle$ satisfies $|\bra{\psi_{2k}}bc\rangle|^2 < \frac s M$, $|\bra{\psi_{2k}}cb\rangle|^2 < \frac s M$, where $s$ is a constant independent of $M,N,n$.
\end{lemma}

\begin{remark}
\textnormal{
    Lemma~\ref{lemma:lemma1} implies that throughout the entire walk process, the amplitude of $|bc\rangle, |cb\rangle$ are always small enough. This ensures the reasonableness of omitting the 3rd and 4th dimensions of the matrix in the proof of Theorem~\ref{th:theorem2}.
}
\end{remark}

Next, we will give a proof of Lemma~\ref{lemma:lemma1}.
\begin{proof}[Proof of Lemma~\ref{lemma:lemma1}]
Let
\begin{equation}\begin{aligned}\label{eq:a12345}
    \left(\begin{matrix}
    a_0^1 \\ a_0^2 \\ a_0^3 \\ a_0^4 \\ a_0^5
    \end{matrix}\right) &= \left(\begin{matrix}
    0 \\ 0 \\ 0\\ 0 \\ 1
    \end{matrix}\right),\\
    \left(\begin{matrix}
    a_i^1 \\ a_i^2 \\ a_i^3 \\ a_i^4 \\ a_i^5
    \end{matrix}\right) &= \e^{\ii \alpha_i}U(\alpha_i,\beta_{i+1})\cdot U(\alpha_i,\beta_i)  \left(\begin{matrix}
    a_{i-1}^1 \\ a_{i-1}^2 \\ a_{i-1}^3 \\ a_{i-1}^4 \\ a_{i-1}^5
    \end{matrix}\right),
\end{aligned}\end{equation}
where the notation $a_i^j$ means the $j$-th row of the state vector after $2i$-th step walk.
And let $\tilde{U}(\alpha)$ be the $3\times 3$ submatrix in the bottom right corner of $\e^{\ii\alpha}U(\alpha,\beta_l) \cdot U(\alpha,\beta_r)$. Then we have
\begin{widetext}\begin{tiny}
\begin{align*}
    \tilde{U}(\alpha) &= \left(
        \begin{array}{ccc}
         1-\frac{\left(1-\e^{-\ii \alpha }\right) (\text{N}-n)}{(M-1) \text{N}} & \frac{\e^{-\ii \alpha } \left(-1+\e^{\ii \alpha }\right)^2 (M-2) (\text{N}-n)}{(M-1)^2 \text{N}} & -\frac{\left(1-\e^{-\ii \alpha }\right) \left(\e^{\ii \alpha }+M-2\right) \sqrt{\frac{(M-2) (\text{N}-n)}{\text{N}}}}{(M-1)^2} \\
         0 & 1-\frac{\left(1-\e^{-\ii \alpha }\right) (\text{N}-n)}{(M-1) \text{N}} & -\frac{\left(1-\e^{-\ii \alpha }\right) \sqrt{\frac{(M-2) (\text{N}-n)}{\text{N}}}}{M-1} \\
         -\frac{\left(1-\e^{-\ii \alpha }\right) \sqrt{\frac{(M-2) (\text{N}-n)}{\text{N}}}}{M-1} & \frac{- \left(1-\e^{-\ii \alpha }\right) \left(\e^{\ii \alpha }+M-2\right) \sqrt{\frac{(M-2) (\text{N}-n)}{\text{N}}}}{(M-1)^2} & \frac{\e^{-\ii \alpha } \left(\e^{\ii \alpha }+M-2\right)^2}{(M-1)^2} \\
        \end{array}
        \right)\\ 
    &\simeq \left(
        \begin{array}{ccc}
         1 & 0 & -(1 - \e^{-\ii \alpha }) \frac{1}{\sqrt{M}} \\
         0 & 1 & -(1 - \e^{-\ii \alpha }) \frac{1}{\sqrt{M}} \\
         -(1 - \e^{-\ii \alpha }) \frac{1}{\sqrt{M}} & -(1 - \e^{-\ii \alpha }) \frac{1}{\sqrt{M}} & \e^{-\ii \alpha }
        \end{array}
        \right).
\end{align*}
\end{tiny}\end{widetext}
Define
\begin{equation}\begin{aligned}\label{eq:P}
    P= \left(\begin{matrix}
        1/\sqrt 2 & 1/\sqrt 2 & 0 \\
        0 & 0 & 1 
        \end{matrix}\right),
\end{aligned}\end{equation}
\begin{equation}\begin{aligned}\label{eq:U'}
    \tilde{U'}(\alpha) &= P \tilde{U}(\alpha) P^\top \\
    &\simeq I - (1-\e^{\ii \alpha}) \left(\begin{matrix}
        \sqrt{\frac 2 M}\\ \sqrt{1 - \frac 2 M}
    \end{matrix}\right)\left(\begin{matrix}
        \sqrt{\frac 2 M}& \sqrt{1 - \frac 2 M}
    \end{matrix}\right),
\end{aligned}\end{equation}
where $\tilde{U'}(\alpha)$ is a 2-dimensional matrix, we can prove through Eq.\eqref{eq:U'} that for arbitrary $\alpha_1, \alpha_2\in \mathbb{R}$ $\tilde{U'}(\alpha)$ satisfies 
\begin{equation}\begin{aligned}\label{eq:U'alpha}
    \tilde{U'}(\alpha_1) \tilde{U'}(\alpha_2) \simeq \tilde{U'}(\alpha_1 + \alpha_2).
\end{aligned}\end{equation}

From Eq.\eqref{eq:U'alpha}, the product of multiple $\tilde{U'}(\alpha)$ can be calculated:
\begin{equation}
    \begin{aligned}
        \prod_{k\ge i \ge 1} \tilde{U'}(\alpha_i) \simeq \tilde{U'}(\sum_{i=1}^k \alpha_i).
    \end{aligned}
\end{equation}

Next, we can compute $a_k^3, a_k^4, a_k^5$ by substituting the expression of $U(\alpha, \beta)$ into Eq.\eqref{eq:a12345}. 
\begin{equation}
\begin{aligned}\label{eq:ak}
    \left(\begin{matrix}
    a_k^3 \\ a_k^4 \\ a_k^5
    \end{matrix}\right) 
    =& \tilde{U}(\alpha_k)\cdot\left(\begin{matrix}
     a_{k-1}^3 \\ a_{k-1}^4 \\ a_{k-1}^5
    \end{matrix}\right) + \left(\begin{matrix}
       \delta_k^3\\\delta_k^4\\\delta_k^5
    \end{matrix}\right) \\
    =& \prod_{k\ge j \ge 1} \tilde{U}(\alpha_j) \cdot \left(\begin{matrix}
       0\\0\\1
    \end{matrix}\right) + \sum_{i=1}^k \prod_{k\ge j \ge i+1} \tilde{U}(\alpha_j) \cdot \left(\begin{matrix}
       \delta_i^3\\\delta_i^4\\\delta_i^5
       \end{matrix}\right),
\end{aligned}
\end{equation}
where the \(\delta\)-related terms in the above equation arise from the elements of the 3-row by 2-column submatrix located in the bottom-left corner of the matrix $\e^{\ii \alpha_i}U(\alpha_i, \beta_{i+1})U(\alpha_i, \beta_{i})$.

For the first half of Eq.\eqref{eq:ak}, as $\tilde{U}(\alpha)$ applying on vector with form $(a,a,b)^\top$ still returns vector with form $(a',a',b')^\top$ , and $P^\top P$ operating on $(a,a,b)^\top$ will not change the vector. Thus we can insert $P^\top P$ items between the multiple of $\tilde{U}(\alpha)$ without changing the result, thus 
\begin{equation}\begin{aligned}\label{eq:insert_P}
    \prod_{k\ge j \ge 1} \tilde{U}(\alpha_j) \cdot \left(\begin{matrix}
        0\\0\\1
     \end{matrix}\right) &= \prod_{k\ge j \ge 1} (P^\top P \tilde{U}(\alpha_j)) \cdot P^\top \left(\begin{matrix}
        0\\1
     \end{matrix}\right)\\ 
     &= P^\top \prod_{k\ge j \ge 1}  (P \tilde{U}(\alpha_j) P^\top) \cdot \left(\begin{matrix}
        0\\1
     \end{matrix}\right)\\ 
     &= P^\top \prod_{k\ge j \ge 1}  \tilde{U'}(\alpha_j) \cdot \left(\begin{matrix}
        0\\1
     \end{matrix}\right)\\ 
     &= P^\top \tilde{U'}(\sum_{j=1}^k \alpha_j) \cdot \left(\begin{matrix}
        0\\1
     \end{matrix}\right)\\ 
     &= \left(\begin{matrix}
        -(1 - \e^{-\ii (\sum_{j=1}^k \alpha_j)  })\frac{1}{\sqrt{M}} \\ -(1 - \e^{-\ii (\sum_{j=1}^k \alpha_j)  }) \frac{1}{\sqrt{M}} \\ \e^{-\ii (\sum_{j=1}^k \alpha_j) }
     \end{matrix}\right),
\end{aligned}\end{equation}
thus
\begin{equation}
    \begin{aligned}\label{eq:leftpart}
        \prod_{k\ge j \ge 1} \tilde{U}(\alpha_j) \cdot \left(\begin{matrix}
            0\\0\\1
         \end{matrix}\right) = \left(\begin{matrix}
            o(1/\sqrt M)\\ o(1/\sqrt M)\\ o(1)
         \end{matrix}\right).
    \end{aligned}
\end{equation}

For the latter half of Eq.\eqref{eq:ak} using a similar method, we can prove that
\begin{equation}
    \begin{aligned}\label{eq:rightpart}
    \sum_{i=1}^k \prod_{k\ge j \ge i+1} \tilde{U}(\alpha_j) \cdot \left(\begin{matrix}
        \delta_i^3\\\delta_i^4\\\delta_i^5
        \end{matrix}\right) = \left(\begin{matrix}
            o(1/\sqrt M)\\ o(1/\sqrt M)\\ o(1)
         \end{matrix}\right).
\end{aligned}
\end{equation}
The proof of Eq.\eqref{eq:rightpart} used the following facts
\begin{equation}
    \begin{aligned}\label{eq:delta}
    \delta_k^3 = o(\sqrt{\frac n {M^2 N}}), 
    \delta_k^4 = o(\sqrt{\frac n {M^2 N}}), 
    \delta_k^5 = o(\sqrt{\frac n {MN}}),
    \end{aligned}
\end{equation}
which is estimated from the elements of the 3-row by 2-column submatrix located in the bottom-left corner of $\e^{\ii\alpha_i}U(\alpha_i, \beta_{i+1})U(\alpha_i, \beta_{i})$.

Combining Eq.\eqref{eq:leftpart} and Eq.\eqref{eq:rightpart} we can get
\begin{equation}
    \begin{aligned}\label{eq:ak2}
        \left(\begin{matrix} 
        a_k^3 \\ a_k^4 \\ a_k^5
        \end{matrix}\right) 
        \simeq \left(\begin{matrix}
           o(1/\sqrt M)\\ o(1/\sqrt M)\\ o(1)
        \end{matrix}\right).
    \end{aligned}
\end{equation}
thus
\begin{align*}
    a_k^3 = o(1/\sqrt{M}),
    a_k^4 = o(1/\sqrt{M}).
\end{align*}
which means there exists a constant $s\in \mathbb{R}$ such that
\begin{align*}
    |\bra{\psi_{2k}}bc\rangle|^2 = |a_k^3|^2 < \frac s M,|\bra{\psi_{2k}}cb\rangle|^2=|a_k^4|^2 < \frac s M.
\end{align*}
\end{proof}

Next, we will give a proof of Theorem~\ref{th:theorem2} using Lemma~\ref{lemma:lemma1}.

\begin{proof}[Proof of Theorem~\ref{th:theorem2}]
From Lemma~\ref{lemma:lemma1} we know that when $M,N$ are large enough, the amplitude of the $3,4$ lines of the state vector are always small
enough. Intuitively, this can be explained as the amplitude of rows 3 and 4 not being amplified by the operator $U(\alpha,\beta)$. Thus we just need to consider $1,2,5$ lines of the state vector. The $3\times 3$ submatrix of $U(\alpha, \beta)$ after dimensionality reduction is 
\begin{small}\begin{align*}
    &U_{sub}(\alpha,\beta) \\=& \left(
\begin{array}{ccc}
 0 & -\e^{-\ii (\alpha -\beta )} & 0 \\
 -1+\frac{\left(1-\e^{-\ii \alpha }\right) n}{(M-1) \text{N}} & 0 & \frac{\left(1+\e^{-\ii \alpha }\right) \sqrt{(M-2) n \text{N}}}{(M-1) \text{N}} \\
 \frac{\left(1+\e^{-\ii \alpha }\right) \sqrt{(M-2) n \text{N}}}{(M-1) \text{N}} & 0 & -1+\frac{\left(1-\e^{-\ii \alpha }\right) (M-2)}{M-1} \\
\end{array}
\right)\\ 
\simeq& \left(
\begin{array}{ccc}
 0 & -\e^{-\ii (\alpha -\beta )} & 0 \\
 -1 & 0 & \frac{\left(1+\e^{-\ii \alpha }\right) \sqrt{(M-2) n \text{N}}}{(M-1) \text{N}} \\
 \frac{\left(1+\e^{-\ii \alpha }\right) \sqrt{(M-2) n \text{N}}}{(M-1) \text{N}} & 0 & -\e^{-\ii \alpha } \\
\end{array}
\right).
\end{align*}\end{small}

We found that the product of two operators $U_{sub}$ is nearly a symmetric matrix. Therefore, we consider analyzing the product of the three-dimensional matrix forms of two search operators, $U_{sub}(\alpha, \beta_l) \cdot U_{sub}(\alpha, \beta_r)$. The product matrix can be approximated as a product of a symmetric matrix and two diagonal matrices:
\begin{widetext}\begin{small}\begin{align*}
&U_{sub}(\alpha, \beta_l) \cdot U_{sub}(\alpha, \beta_r) \\ 
=& \left(
\begin{array}{ccc}
 \e^{-\ii (\alpha -\beta_l )} & 0 & -\frac{\left(-1+\e^{\ii \alpha }\right) \e^{\ii \beta_l -2 i \alpha } \sqrt{(M-2) n \text{N}}}{(M-1) \text{N}} \\
 \frac{\e^{-2 \ii \alpha } \left(-1+\e^{\ii \alpha }\right)^2 (M-2) n}{(M-1)^2 \text{N}} & \e^{-\ii (\alpha -\beta_r )} & -\frac{\e^{-2 \ii \alpha } \left(-1+\e^{\ii \alpha }\right) \sqrt{(M-2) n \text{N}}}{(M-1) \text{N}} \\
 -\frac{\e^{-2 \ii \alpha } \left(-1+\e^{\ii \alpha }\right) \sqrt{(M-2) n \text{N}}}{(M-1) \text{N}} & -\frac{\left(-1+\e^{\ii \alpha }\right) \e^{\ii \beta_r -2 i \alpha } \sqrt{(M-2) n \text{N}}}{(M-1) \text{N}} & \e^{-2 \ii \alpha } \\
\end{array}
\right) \\ 
\simeq& \left(
\begin{array}{ccc}
 \e^{-\ii (\alpha -\beta_l )} & 0 & -\frac{\left(-1+\e^{\ii \alpha }\right) \e^{\ii \beta_l -2 i \alpha } \sqrt{(M-2) n \text{N}}}{(M-1) \text{N}} \\
0 & \e^{-\ii (\alpha -\beta_r )} & -\frac{\e^{-2 \ii \alpha } \left(-1+\e^{\ii \alpha }\right) \sqrt{(M-2) n \text{N}}}{(M-1) \text{N}} \\
 -\frac{\e^{-2 \ii \alpha } \left(-1+\e^{\ii \alpha }\right) \sqrt{(M-2) n \text{N}}}{(M-1) \text{N}} & -\frac{\left(-1+\e^{\ii \alpha }\right) \e^{\ii \beta_r -2 i \alpha } \sqrt{(M-2) n \text{N}}}{(M-1) \text{N}} & \e^{-2 \ii \alpha } \\
\end{array}
\right)\\
=&
\e^{-\ii\alpha}
D_1(\beta_l)\cdot  A(\alpha) \cdot
D_2(\beta_r),
\end{align*}
\end{small}\end{widetext}
where $\mu =  \frac{\sqrt{(M-2) n \text{N}}}{(M-1) \text{N}}$, and
\begin{align*}
    D_2(\beta_r) = \left(\begin{array}{ccc} 
         1 & 0 & 0 \\
         0 & \e^{\ii \beta_r } & 0 \\
         0 & 0 & 1 \\
        \end{array}\right), 
    D_1(\beta_l) = \left(\begin{array}{ccc}
        \e^{\ii \beta_l} & 0 & 0 \\
        0 & 1 & 0 \\
        0 & 0 & 1 \\
    \end{array}\right),
\end{align*}
\begin{align*}
    A(\alpha) = \left(
    \begin{array}{ccc}
     1 & 0 & \mu(\e^{-\ii \alpha }-1) \\
     0 & 1 & \mu(\e^{-\ii \alpha }-1) \\
     \mu(\e^{-\ii \alpha }-1)  & \mu(\e^{-\ii \alpha }-1) & \e^{-\ii \alpha } \\
    \end{array}
    \right).
\end{align*}

Applying a series of operators $U_{sub}(\alpha_i,\beta_{i+1}) U_{sub}(\alpha_i,\beta_i)(i = 1,...,t)$ on the initial state $\ket{\bar 0}_3$ we can get
\begin{widetext}\begin{align*}
    &\prod_{t\ge i\ge 1} U_{sub}(\alpha_i, \beta_{i+1}) \cdot U_{sub}(\alpha_i, \beta_i) \ket{\bar 0}_3 \\ 
    &= \prod_{t\ge i\ge 1} (\e^{-\ii\alpha_i}
        D_1(\beta_{i+1})\cdot 
        A(\alpha_i)
        \cdot
        D_2(\beta_{i})) \ket{\bar 0}_3 \\
    &= \left(\prod_{t\ge i\ge 1} \e^{-\ii\alpha_i}\right)   
        D_1(\beta_{t+1}) \left( \prod_{t\ge i\ge 1} ( 
        A(\alpha_i) \cdot D_2(\beta_{i}) \cdot D_1(\beta_i)) \right) (D_1(-\beta_{1})\ket{\bar 0}_3 )\\
    &= \left(\prod_{t\ge i\ge 1} \e^{-\ii\alpha_i}\right)   
        D_1(\beta_{t+1}) \left( \prod_{t\ge i\ge 1} ( 
        A(\alpha_i) \cdot D_2(\beta_{i}) \cdot D_1(\beta_i)) \right) \ket{\bar 0}_3,
\end{align*}\end{widetext}
where 
\begin{align*}
    \left(\prod_{t\ge i\ge 1} \e^{-\ii\alpha_i}\right)   
        D_1(\beta_{t+1})
\end{align*}
doesn't change the probability of measurement, so we only need to analyze the remaining part.

Let $D_{12}(\beta) = D_1(\beta)\cdot D_2(\beta)$,
\begin{align*}
    A(\alpha) \cdot D_2(\beta) \cdot D_1(\beta) = A(\alpha) \cdot D_{12}(\beta).
\end{align*}

Since the operator above acts on vectors of the form $(a,a,b)^\top$ to yield vectors of the form $(a',a',b')^\top$, and the initial state is $\ket{\bar{0}}_3 = (0,0,1)^\top$, which conforms to the form $(a,a,b)^\top$, we can use a similar method as in the proof of Lemma~\ref{lemma:lemma1} to project and reduce the dimension of the operator. Let
\begin{align*}
    C'_{sub}(\alpha) &= P\cdot A(\alpha) \cdot P^\top \\ 
    &= \left(\begin{array}{cc}
        1 & \sqrt 2\mu(\e^{-\ii \alpha }-1) \\
        \sqrt 2\mu(\e^{-\ii \alpha }-1) & \e^{-\ii \alpha } \\
        \end{array}\right), \\
    Q'_{sub}(\beta) &=P\cdot D_{12}(\beta) \cdot P^\top\\ 
    &= \left(
        \begin{array}{cc}
         \e^{\ii \beta } & 0 \\
         0 & 1 \\
        \end{array}
        \right),
\end{align*}
where $P$ is defined by Eq.\eqref{eq:P}.

Using the similar method in the proof of Lemma~\ref{lemma:lemma1}, we can compute the three-dimensional state vector after $2t$ steps of the robust quantum walk search as follows:
\begin{widetext}\begin{equation}\begin{aligned}
    \label{eq:psi_2t}
    \ket{\psi_{2t}}_3 &= \left(\prod_{t\ge i\ge 1} \e^{-\ii\alpha_i}\right)   
        D_1(\beta_{t+1}) \left( \prod_{t\ge i\ge 1} A(\alpha_i) \cdot D_{12}(\beta_{i}) \right) \ket{\bar 0}_3\\ 
    &= \left(\prod_{t\ge i\ge 1} \e^{-\ii\alpha_i}\right)   
        D_1(\beta_{t+1}) P^\top \left( \prod_{t\ge i\ge 1} P( 
        A(\alpha_i) \cdot P^\top P \cdot D_{12}(\beta_{i}))P^\top \right) P\ket{\bar 0}_3\\
    &= \left(\prod_{t\ge i\ge 1} \e^{-\ii\alpha_i}\right)   
        D_1(\beta_{t+1}) P^\top \left( \prod_{t\ge i\ge 1} C'_{sub}(\alpha)Q'_{sub}(\beta) \right) \ket{\bar 0}_2.
\end{aligned}\end{equation}\end{widetext}
The reason why the second ``=" of Eq.\eqref{eq:psi_2t} holds is similar to that of Eq.\eqref{eq:insert_P}.

From Ref.\cite{yoder_fixed-point_2014} we know that for $2\times2$ matrix $C'_{sub}(\alpha)Q'_{sub}(\beta)$, we have 
\begin{align*}
|\bra{\bar 0}_2\left( \prod_{t\ge i\ge 1} C'_{sub}(\alpha_i)Q'_{sub}(\beta_i) \right) \ket{\bar 0}_2|^2
    \le \epsilon.
\end{align*}
Thus the success probability of finding a marked vertex after $2t$ steps is
\begin{align*}
    p_{succ}(2t) =& 1 - |\langle \psi_{2t}| cc \rangle|^2 -  |\langle \psi_{2t}| cb \rangle|^2 -  |\langle \psi_{2t}| bc \rangle|^2\\ 
    \simeq& 1 - |\bra{\psi_{2t}}_3 \bar 0\rangle_3|^2\\ 
    =& 1 - |\bra{\bar 0}_2\left( \prod_{t\ge i\ge 1} C'_{sub}(\alpha)Q'_{sub}(\beta) \right) \ket{\bar 0}_2|^2\\ 
    \ge & 1 - \epsilon.
\end{align*}

This accomplishes the proof of Theorem~\ref{th:theorem2}.
\end{proof}

\section{Simulation}\label{sec:simulation}
This section presents numerical simulation results of quantum walk search algorithms on complete multipartite graphs. We demonstrate the simulation results for standard quantum walks and robust quantum walks, categorized into non-robust and robust algorithm cases. Moreover, we analyze these results and discuss their implications, aiming to validate the theoretical analyses presented in section \ref{sec:robust}.

Figure~\ref{fig:robust,case1,2} presents the simulation result of the first case, the subfigures have the same parameters $M=1000,N=10000,\epsilon=0.1$ but different numbers of marked vertices $n$, respectively $1$ and $10$. The horizontal axis represents the number of steps in the quantum walk, while the vertical axis represents the probability of finding the marked vertices. According to the simulation results, it is observed that beyond a certain number of steps, the final success probability oscillates within the range of $(1-\epsilon^2, 1)$. Similarly, Figure~\ref{fig:robust,case2,2} shows the result of case 2, demonstrating that the final success probability oscillates within the range of $(1-\epsilon, 1)$.

\begin{figure}[htbp]
    \centering
    \begin{subfigure}{0.4\textwidth}
        \includegraphics[width=\linewidth]{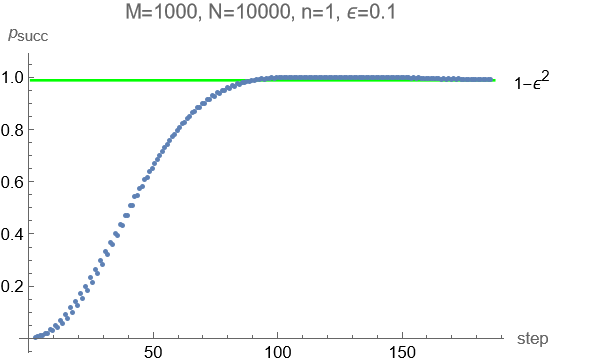}
    \end{subfigure}
    
    \begin{subfigure}{0.4\textwidth}
        \includegraphics[width=\linewidth]{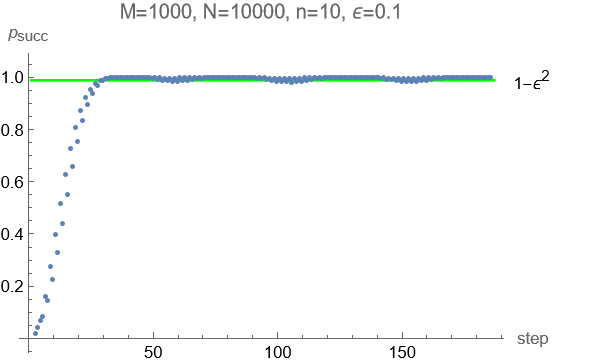}
    \end{subfigure}
    \caption{Robust quantum walk of case 1.}
    \label{fig:robust,case1,2}
\end{figure}

\begin{figure}[htbp]
    \centering
    \begin{subfigure}{0.4\textwidth}
        \includegraphics[width=\linewidth]{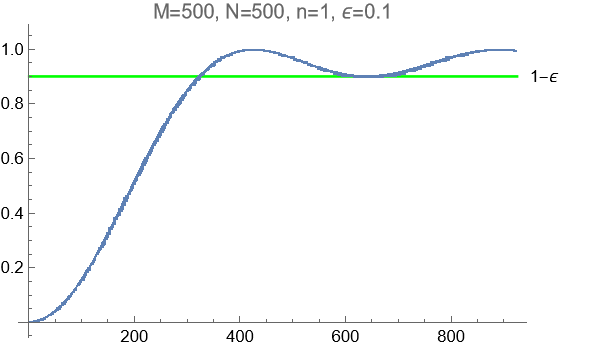}
    \end{subfigure}\hfill
    \begin{subfigure}{0.4\textwidth}
        \includegraphics[width=\linewidth]{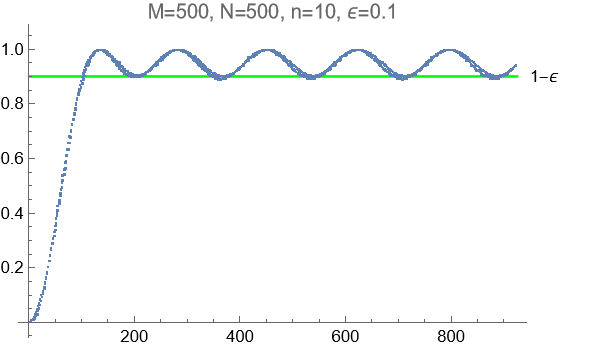}
    \end{subfigure}
    \caption{Robust quantum walk of case 2.}
    \label{fig:robust,case2,2}
\end{figure}

\section{Quantum Circuit Implementation}\label{sec: circuit}

In this section, we give the quantum circuit implementation of our algorithm. 
We assume $l_m = \lceil \log_2(M-1)\rceil, l_n = \lceil\log_2(N)\rceil$ for simplicity. If $M-1$ or $N$ is not a power of 2, we can add some imaginary vertices and edges into the graph. Then we can use $l_n$ qubits to encode the index of a vertex in a subset, and $l_m+1$ qubits to encode $M$ subsets. Thus we can use $l_m + l_n + 1$ qubits to represent a vertex in a complete multipartite graph, and $l_m + l_n$ qubits to represent a direction to go in the next step.

\begin{figure}
    \centering
    \includegraphics[width=\linewidth]{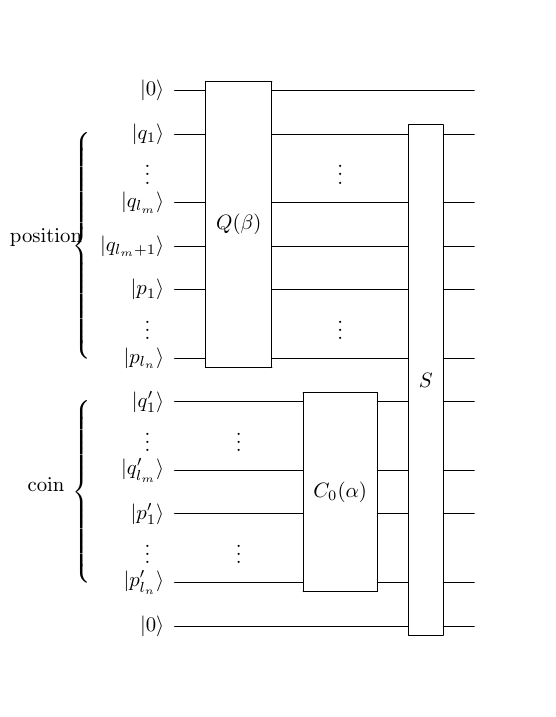}
    \caption{One step of robust quantum walk.}
    \label{fig:onestep} 
    \label{fig:enter-label}
\end{figure}

One step of the robust quantum walk is shown in Figure~\ref{fig:onestep}. $\ket{q_1...q_{l_m + 1}}$ represents the subset index of vertex, and $\ket{p_1...p_{l_n}}$ represents the index of vertex in one subset, these $l_m + l_n + 1$ qubits composed the position register. As for the coin register, $\ket{p'_1...p'_{l_n}}$ is the index of the vertex in a subset, but the $q'$ registers are a little complicated. For a vertex $v$ in a complete multipartite graph, there are $(M-1)N$ neighbors in $(M-1)$ subsets. Thus $\ket{q'_1...q'_{l_m}}$ means the index of the subset that contains the neighbors, which means we skipped the subset that contains vertex $v$ when indexing the coin register. 
For example, if the position register points to the $p_u$-th vertex of the $q_u$-th subset,  the coin register points to the $p_v$-th vertex of the $q_v$-th subset, the encode of this state is
\begin{align*}
    \begin{cases}
        \ket{q_u p_u, q_v p_v} & q_u > q_v,\\
        \ket{q_u p_u, (q_v-1) p_v} & q_u < q_v.
    \end{cases}
\end{align*}
As for the gates in this graph, $S$ is the shift operator, $C_0(\alpha)$ is the coin operator in this encoding, and $Q(\beta)$ is the query operator operating on the position register.

Operator $S$ will exchange the vertex in the position register and the coin register. Assuming we have a encoded state $\ket{q_u p_u, q_v p_v}$. According to the encoding rule, if $q_u \le q_v$, the corresponding subset index of $q_v$ is $q_v + 1$, then after exchanging them, the corresponding subset index of two registers are $q_v + 1$ and $q_u$, as $q_v + 1 > q_u$, the encoding of $q_u$ is still $q_u$. If $q_u > q_v$, the corresponding subset index of $q_v$ is still $q_v$, but when index after exchange, as $q_v < q_u$, the encode of $q_u$ is $q_u-1$. For the above reasons, $S$ acting on state $\ket{q_u p_u, q_v p_v}$ will output
\begin{align*}
    S \ket{q_u p_u, q_v p_v} = \begin{cases}
        \ket{(q_v+1) p_v, q_u p_u}& q_u \le q_v,\\
        \ket{q_v p_v, (q_u - 1) p_u}& q_u > q_v.
    \end{cases}
\end{align*}

The shift operator $S$ can be realized in this way. First, we exchange $\ket{p_v}$ and $\ket{p_u}$. Then we compare $q_u$ and $q_v$, add $1$ if $q_u \le q_v$, and subtract $1$ otherwise. Next, we exchange these two terms. At last, we make a comparison again to recover the auxiliary qubit to $\ket{0}$. In circuit design, to prevent overflow in addition, we perform the addition operation after swapping. The quantum circuit design of $S$ is shown in Figure~\ref{fig:S operator}. As shown in the graph, part \normalsize{\textcircled{\scriptsize{1}}}\normalsize\enspace is the circuit to exchange $p_v$ and $p_u$, and part \normalsize{\textcircled{\scriptsize{2}}}\normalsize\enspace is the circuit to process the qubits representing the index of subset in these two registers and to recover the auxiliary qubit. 
Figure~\ref{fig:compose_S} shows two operations to implement operator $S$. 

\begin{figure}[htbp]
    \centering
    \begin{subfigure}{0.3\textwidth}
        \includegraphics[width=\textwidth]{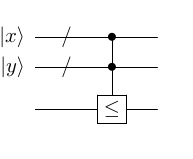}
        \caption{}
        \label{fig:le}
    \end{subfigure}
    \hfill
    \begin{subfigure}{0.3\textwidth}
        \includegraphics[width=\textwidth]{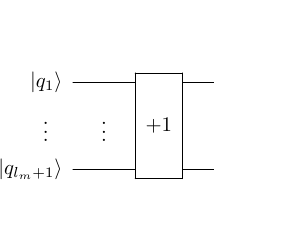}
        \label{fig:add1}
        \caption{}
    \end{subfigure}
    \caption{Two crucial circuit modules. (a)Compare two registers, and set the third register to $\ket{1}$ if $x \le y$, to $\ket{0}$ otherwise. If $x$ and $y$ has $l_m$ qubits, this gate can be composed of $O(l_m)$ fundamental gates\cite{oliveira2007quantum}. (b)Add $1$ to the register, this gate needs $O(l_m)$ basic gates to realize\cite{cuccaro2004new}. The $-1$ gate is similar.}
    \label{fig:compose_S}
\end{figure}

Recall that
\begin{align*}    C(\alpha)=\sum_u\left|u\right>\left<u\right|\otimes\left[\left(1-\e^{-\ii\alpha}\right)\left|s_u\right>\left<s_u\right|-I\right].
\end{align*}
But in our encoding, the vector $\ket{s_u}$ can be treated as the uniform superposition on the coin register $H^{\otimes (l_m + l_n)}\left|0\right>$ which is not related to $u$. So we can write the encoded coin operator as
\begin{align*}
    &\sum_u\left|u\right>\left<u\right|\otimes \left[\left(1-\e^{-\ii\alpha}\right)H^{\otimes (l_m + l_n)}\left|0\right>\left<0\right|H^{\otimes (l_m + l_n)}-I\right] \\ 
    =& I \otimes \left[ H^{\otimes (l_m + l_n)} ((1 - e^{-\ii \alpha})|0\rangle\langle0| - I) H^{\otimes (l_m + l_n)} \right]\\
    =& I \otimes C_0(\alpha).
\end{align*}
Then $C_0(\alpha)$ can be realized in Figure~\ref{fig:C operator}.
 
\begin{figure*}[htbp]
    \centering
    \begin{subfigure}{0.4\textwidth}
        \includegraphics[width=\linewidth]{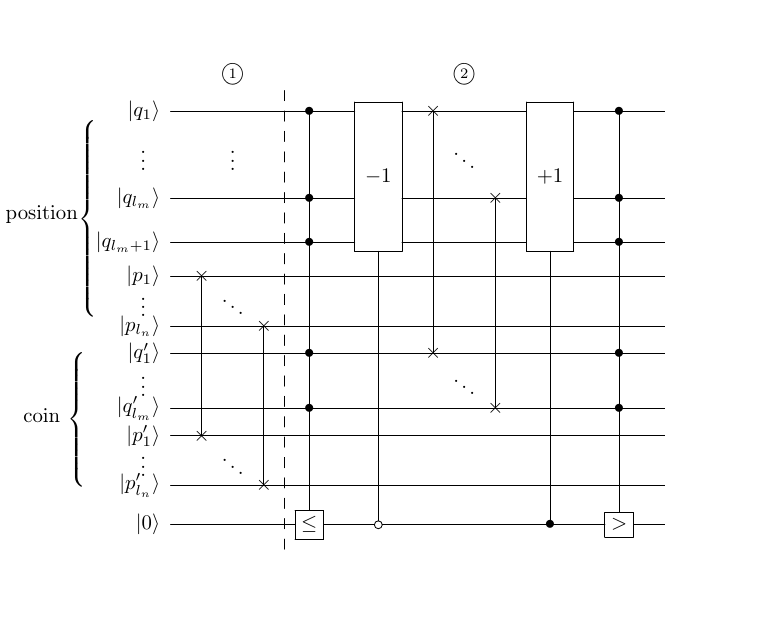}
        \caption{}
        \label{fig:S operator} 
    \end{subfigure}
    \hfill 
    \begin{subfigure}{0.29\textwidth}
        \includegraphics[width=\linewidth]{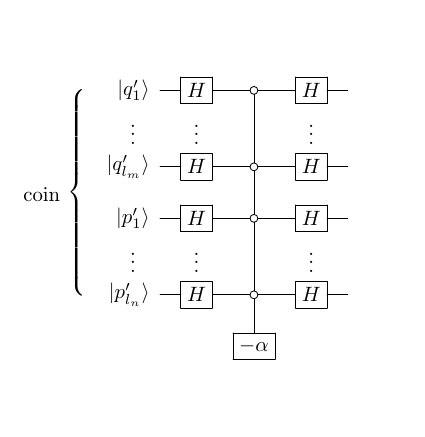}
        \caption{}
        \label{fig:C operator}
    \end{subfigure}
    \hfill 
    \begin{subfigure}{0.29\textwidth}
        \includegraphics[width=\linewidth]{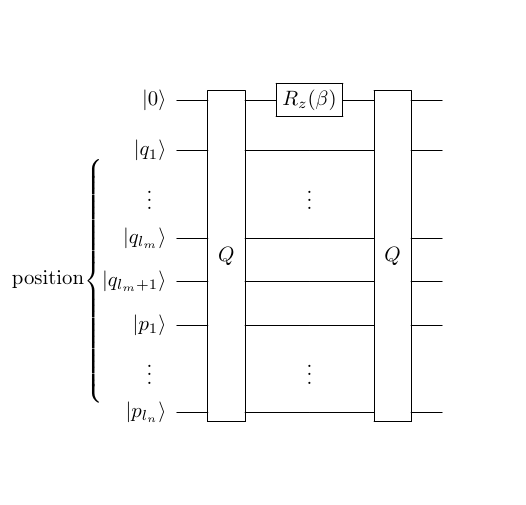}
        \caption{}
        \label{fig:Q operator}
    \end{subfigure}
    \caption{Quantum Circuit of $S,C_0(\alpha)$ and $Q(\beta)$. (a)Quantum Circuit of $S$. (b)Quantum Circuit of $C_0(\alpha)$. (c)Quantum Circuit of $Q(\beta)$. The operator $Q$ is the query operator. If the vertex in the position register is marked, then the auxiliary qubit will be set to $\ket{1}$.}
\end{figure*}

In addition, the query operator $Q(\beta)$ can be realized in Figure~\ref{fig:Q operator}.
Operator $Q$ is the oracle that distinguishes marked vertices from unmarked ones. $R_z(\beta)$ represents a rotation about the $z$ axis by angle $\beta$, and it's matrix form is
\begin{align*}
    R_z(\beta) = \left(\begin{matrix}
        1 & 0 \\ 0 & \e^{\ii \beta}
    \end{matrix}\right).
\end{align*}
The operator $Q(\beta)$ is implemented in the following steps. First, apply the operator $Q$ to set the auxiliary qubit to $\ket{1}$ if the vertex is marked, and to $\ket{0}$ otherwise. Next, a rotation operator is applied to the auxiliary qubit. Finally, apply the operator $Q$ again to reset the auxiliary qubit back to $\ket{0}$. 

The quantum circuit of one step quantum walk search is composed of $S,C_0(\alpha)$ and $Q(\beta)$ as shown in Figure~\ref{fig:onestep}. The number of qubits we need is $O(l_m + l_n) = O(\log(MN))$ which is the logarithm of the total number of vertices. 
The operator $S$ needs $l_m + l_n$ swap gates, two $l_m$ qubit compare gates, and two add or minus gates. These swap gates need $3(l_m + l_n)$ basic gates, the two compare gates needs $O(l_m)$ basic gates, the add gate and minus gate needs $O(l_m)$ basic gates, so it needs $3(l_m + l_n) + O(l_m) + O(l_m) = O(l_m + l_n) = O(\log(MN))$ gates in total to compose $S$ gate.
$C_0(\alpha)$ needs a $l_m + l_n$ qubits controlled phase shift gate and $2(l_m + l_n)$ Hadamard gates. The controlled phase gate takes $O((l_m + l_n)^2)$ basic gates\cite{saeedi2013linear}. Thus $C(\alpha)$ takes $2(l_m + l_n) +O((l_m + l_n)^2) = O(\log^2(MN))$ basic gates. 
The number of basic gates of $Q(\beta)$ needs depends on the query operator $Q$ noted as $C_Q$. 
According to the analysis above, the total gate number of one step is $\max(O(\log^2(MN)), C_Q)$.  

\section{Conclusion}\label{sec:conclusion}
In this paper, we first applied the coined quantum walk model to complete M-partite graph with multiple marked vertices in two cases and found a quadratic speedup over the classical search. 
Furthermore, we meticulously crafted robust quantum walk search algorithms tailored to two distinct scenarios involving marked vertices.  For both instances, we rigorously demonstrated that the success probability of locating the marked vertices approaches $1$ with negligible deviation, contingent upon executing a sufficiently large number of walk steps.  This pivotal finding underscores our ability to search the marked vertices with a high degree of certainty, even without prior knowledge of their exact count.  Remarkably, our algorithm not only achieves this but also retains a quadratic speedup over its classical counterparts, significantly enhancing search efficiency.  Additionally, we offer a comprehensive circuit implementation of our innovative algorithm, facilitating its practical application and further exploration.

There are more valuable research directions related to this problem that can be explored in the future.
For example, we can also study the robustness of the quantum walk search algorithm on complete M-partite graph with different numbers of vertices or different numbers of marked vertices in each set. Also, the robustness of the quantum walk search algorithm on other types of graphs is worth investigating. We will try to address these issues in future work.

\begin{acknowledgments}
This work was supported in part by the National Natural Science Foundation of China Grants No. 62325210, 62301531, and China Postdoctoral Science Foundation Grant No. 2022M723209.
\end{acknowledgments}

\bibliographystyle{apsrev4-2} 
\bibliography{quantumwalk}
\end{document}